\documentclass[journal,onecolumn]{IEEEtran}

\usepackage{amssymb,mathtools}
\usepackage{array}
\usepackage{xcolor}
\usepackage{amsmath,amssymb,amsthm,mathtools}
\usepackage{cite}
\usepackage{array}
\usepackage{url}
\usepackage{xcolor}

\newtheorem{theorem}{Theorem}
\newtheorem{lemma}{Lemma}
\newtheorem*{conjecture}{Conjecture}
\newtheorem{proposition}{Proposition}

\newtheorem{corollary}{Corollary}

\newcommand{\C}{\mathcal{C}}
\newcommand{\F}{\mathbb{F}}
\newcommand{\wt}{\mathrm{wt}}
\newcommand{\wtq}{\mathrm{wt}_q}
\newcommand{\cl}{\operatorname{cl}}
\newcommand{\GRM}{\mathrm{GRM}}
\newcommand{\PGRM}{\mathrm{PGRM}}

\begin{document}

\title{Counterexamples to Charpin's Conjecture on BCH codes}
\author{Run Zheng, Yaoran Yang, Yutong Zhang, Haode Yan, and Maosheng Xiong%
\thanks{Run Zheng and Maosheng Xiong are with the Department of Mathematics,
The Hong Kong University of Science and Technology, Hong Kong, China
(e-mail: zhengrung@ust.hk; mamsxiong@ust.hk).}%
\thanks{Yaoran Yang and Yutong Zhang are with the School of Mathematics,
Sichuan University, Chengdu 610065, China
(e-mail: yangyaoran@stu.scu.edu.cn; yutongzhang@stu.scu.edu.cn).}
\thanks{Haode Yan is with the School of Science, Harbin Institute of Technology,
Shenzhen 518055, China (Corresponding author)
(e-mail: yanhd@hit.edu.cn).} }

\maketitle

    \begin{abstract}
We construct an infinite family of $q$-ary primitive narrow-sense BCH codes whose minimum distance strictly exceeds the Bose distance; in fact, the gap between the two can be arbitrarily large as the length of the code tends to infinity. The key idea is to embed these BCH codes in a suitably large punctured generalized Reed--Muller code, whose codeword weights obey divisibility conditions supplied by Ax's theorem. This divisibility forces the minimum distance of the BCH codes far above the Bose distance.
In particular, our family disproves a longstanding conjecture of Charpin asserting that this difference is at most four.
\end{abstract}


\section{Introduction}
Bose-Chaudhuri-Hocquenghem (BCH) codes were introduced
by Hocquenghem \cite{Hocquenghem1959} and independently by Bose and Ray-Chaudhuri \cite{BoseRayChaudhuri1960} more than six decades ago, and they remain one of the most important families of cyclic codes. Their simple algebraic structure enables efficient encoding and decoding while providing strong error-correcting capabilities, and this combination has made them ubiquitous in practice, from communication systems to data storage devices and consumer electronics. 

Despite their long history, practical importance, and simple algebraic structure, BCH codes are far from fully understood. In particular, their minimum distances are still not known explicitly in general. The
difficulty of this problem was already singled out in the classical text of MacWilliams and Sloane \cite[Ch.~9]{MacWilliamsSloane1977}, and it has been
emphasized again in the surveys of Charpin \cite{Charpin1998} and of Ding and
Li \cite{DingLi2024}. Even for primitive narrow-sense BCH codes, the best known and most studied
subfamily, the minimum distance has been determined only when the parameters lie in certain ranges, and it remains unknown outside them.

Throughout this paper, 
let $q$ be a prime power and $m\geq 1$ an integer. Let $\beta$ be a primitive element of the finite field $\F_{q^m}$ of order $q^m$. For each positive integer $i$, denote by $m_{i}(x)$ the minimal polynomial of $\beta^i$ over $\F_q$. The $q$-ary \emph{primitive narrow-sense}
BCH code of length $n=q^m-1$ with \emph{designed distance} $\delta $ where $2\leq \delta \le n$, denoted by 
$\mathcal{C}_{(q,m,\delta)}$, is the cyclic code of length $n$ with the generator polynomial $g_{(q,m,\delta)}(x)$ given by 
\begin{equation*}
    g_{(q,m,\delta)}(x)=\mathrm{lcm}\left(m_{1}(x), m_{2}(x),\ldots,m_{\delta-1}(x)\right), 
    \end{equation*}
where  $\mathrm{lcm}$ denotes the least common multiple of the polynomials.  

By the BCH bound, $d\bigl(\mathcal{C}_{(q,m,\delta)}\bigr)\geq \delta$. The first subtlety, however, is that different designed distances may define the same code; that is, it may happen that
$
  \mathcal{C}_{(q,m,\delta)}
  =\mathcal{C}_{(q,m,\delta')},
  $
where $\delta' > \delta$. The largest such $\delta'$ that defines $\mathcal{C}_{(q,m,\delta)}$ is called the \emph{Bose distance} of $\mathcal{C}_{(q,m,\delta)}$, denoted by $d_B\bigl(\mathcal{C}_{(q,m,\delta)}\bigr)$. So we have
\begin{equation*}
  d\bigl(\mathcal{C}_{(q,m,\delta)}\bigr)
  \geq d_B\bigl(\mathcal{C}_{(q,m,\delta)}\bigr)
  \geq \delta.
\end{equation*}
When the code is clear from the context, we simply write $d$ and $d_B$ for
its minimum distance and Bose distance, respectively.

For fixed $q$, $m$ and $\delta$, determining $d_B$ explicitly is already a
nontrivial problem, and it has attracted considerable attention in its own
right \cite{D2015,DFZ2017,YF2000,YH1996,zheng2025,zheng2026}. Understanding the minimum
distance itself therefore amounts to understanding the gap $d-d_B$, and a
substantial body of evidence suggests that for primitive narrow-sense BCH
codes this gap vanishes more often than not. Peterson \cite{Peterson1969}
proved that a primitive narrow-sense BCH code with Bose distance $q^t-1$ has
minimum distance exactly $q^t-1$ for $1\le t\le m-1$. Berlekamp
\cite{Berlekamp1970} and Kasami and Lin \cite{KasamiLin1972} exhibited further
families satisfying $d=d_B$, Augot and Sendrier \cite{AugotSendrier1994} added
more in the 1990s, and this line of work has continued in recent years
\cite{ChenEtAlBCH2026,ChenEtAlGoppaBCH2026,Li2017,NoguchiEtAl2021,
ShanyBerman2025,TiwariKewat2026}. We also mention the general result of Cohen \cite[Theorem 1.2]{MR1431812} showing that for binary primitive narrow-sense BCH codes $\mathcal{C}_{(2,m,2t+1)}$, one has $d=d_B=2t+1$ whenever $m >\log_2\left((2t+1)! \cdot (2t-1)\right)$. By contrast, primitive narrow-sense BCH
codes with $d>d_B$ are much harder to find: Kasami and Tokura
\cite{KT1969} exhibited binary examples satisfying $d\ge d_B+2$, but to the
best of our knowledge all previously known infinite family whose exact
minimum distance was determined satisfies $d=d_B$.

Extensive numerical
computations for binary BCH codes of lengths up to $511$
\cite{AugotCharpinSendrier1992,CanteautChabaud1998} show that the gap between $d$ and $d_B$ is small. Charpin wrote that ``When the codes are binary and primitive, it is usually conjectured that the true minimum distance $d$ does not exceed $\delta+4$, $\delta$ being the designed (Bose)
distance.'' (see \cite[page 60]{Charpin1998}).
A more general version, asserting the same bound for every $q$-ary primitive narrow-sense BCH code $\mathcal{C}_{(q,m,\delta)}$, was subsequently formulated in
\cite{D2015,DFZ2017,Li2017} as \emph{Charpin's conjecture} and is stated below.
\begin{conjecture}
Let $q$ be a prime power. 
For any integers  $m$ and $\delta$ with 
$m\geq 2$ and $2\leq\delta\leq q^m-1$, we have
\begin{equation}\notag
d\bigl(\mathcal{C}_{(q,m,\delta)}\bigr)
-d_B\bigl(\mathcal{C}_{(q,m,\delta)}\bigr)
\leq 4.
\end{equation}
\end{conjecture}

In this paper, we disprove this conjecture in a strong sense. Our main result is the following.
\begin{theorem}\label{mainth}
Let $q$ be a prime power, and let  $m$ be an integer with $m\ge 13$.  Set $u=\lfloor m/4\rfloor$ and
$t=\lfloor (m-1)/3\rfloor$. Let $s$ be an integer with $u\le s<t$, and let
$\delta=q^m-q^{m-1}-q^{m-1-u}-q^{s}-1$. Then
\begin{equation} \label{1:gapb}
d_B(\mathcal{C}_{(q,m,\delta)})=\delta\quad \hbox{ and }\quad 
  d(\mathcal{C}_{(q,m,\delta)}) \geq \delta+q^{s}, 
\end{equation}
with equality in the latter holding when $q = 2$.
\end{theorem}

Theorem \ref{mainth} states that the primitive narrow-sense BCH codes in this family satisfy
\begin{equation*} 
    d-d_B\geq q^{\lfloor m/4 \rfloor},
\end{equation*}
that is, the gap between the true minimum distance and the Bose distance of a primitive narrow-sense $q$-ary BCH code can be arbitrarily large, contrary to Charpin's conjecture that the gap is bounded by four\footnote{While writing this paper, we realize that, for $q=2$, $m=13$, $u=s=3$, $t=4$, and $\delta=3575$, the code $\mathcal{C}_{(q,m,\delta)}$ is exactly the code appearing in the last line of \cite[Table I]{KT1969}, where it was shown that $d-d_B\geq 8$. This reference seems to be overlooked by the coding theory community and by us until now.}. For the sake of completeness, we also note that the code $\mathcal{C}_{(q,m,\delta)}$ in Theorem \ref{mainth} has dimension 
\begin{equation}\label{dimeq}
    \mathrm{dim}(\mathcal{C}_{(q,m,\delta)})\!=\!1+ m+\frac{m(m-2u+1)}{2}
+\frac{m}{3}
\left(
\binom{m-3u-1}{2}+3(s-u+1)
\right).
\end{equation}

It is worth noting that even the smallest codes in this family have large
lengths and designed distances; see Table~\ref{table1}. Exact
minimum-distance computations for codes with such parameters are beyond
the range of routine direct computation. This may explain why this family has   escaped the attention of the coding-theory
community.

\begin{table}[htbp]
\centering
\caption{Parameters of some binary primitive narrow-sense BCH codes in
Theorem~\ref{mainth}.}
\label{table1}
\begin{tabular}{cccccc}
\hline
$m$ & $n=2^m-1$ & $\dim$ & $d_B$ & $d$ & $d-d_B$ \\
\hline
13 & 8191   & 92  & 3575  & 3583  & 8  \\
14 & 16383  & 120 & 7159  & 7167  & 8  \\
15 & 32767  & 156 & 14327 & 14335 & 8  \\
16 & 65535  & 121 & 30703 & 30719 & 16 \\
17 & 131071 & 154 & 61423 & 61439 & 16 \\
\hline
\end{tabular}
\end{table}

The idea of the proof of Theorem \ref{mainth} is very simple. First, it is straightforward to verify that the Bose distance of $\mathcal{C}_{(q,m,\delta)}$ is exactly $\delta$ by demonstrating that $\delta$ is a coset leader, so $d \ge \delta$. The code dimension is also easily obtained. Second, this is the technical part: we embed $\mathcal{C}_{(q,m,\delta)}$ inside a suitable punctured generalized Reed--Muller code $\PGRM_q(\ell,m)$, for which the weight of each codeword satisfies strong divisibility conditions (see Corollary \ref{lemma2}) by applying Ax's theorem. Then this divisibility naturally forces $d$ far away from $d_B=\delta$, to the effect that $d \ge \delta+q^s$. We also note that this idea was already present in \cite{KT1969}, though here we apply it in a much more sophisticated way in order to construct an infinite family of such BCH codes satisfying \eqref{1:gapb}.

The remainder of this paper is organized as follows.
Section~\ref{preliminaries} introduces the notation and collects the
preliminary results on primitive narrow-sense BCH codes,
$q$-cyclotomic cosets, and punctured generalized Reed--Muller codes
needed in the sequel. Section \ref{sec:auxiliary} establishes several auxiliary lemmas, which
are then applied in Section~\ref{mainresult} to prove the main theorem.
Finally, Section~\ref{conclusion} concludes the paper with some remarks
and open questions. The computation of the dimension of the code $\mathcal{C}_{(q,m,\delta)}$ is left to Section~\ref{append} {\bf Appendix}. 

\section{Preliminaries}\label{preliminaries}

\subsection{Cyclic codes}
Let  $\mathbb{F}_q^n$ denote the $n$-dimensional vector space over $\mathbb{F}_q$.
A $q$-ary linear code is a linear subspace of $\mathbb{F}_q^n$. In particular, a linear code $\mathcal{C} \subseteq \mathbb{F}_q^n$ is said to be \textit{cyclic} if $(c_0, c_1, \dots, c_{n-1}) \in \mathcal{C}$ implies that its cyclic shift $(c_{n-1}, c_0, \dots, c_{n-2}) \in \mathcal{C}$. By identifying each vector $(c_0, c_1, \dots, c_{n-1}) \in \mathbb{F}_q^n$ with its polynomial representation
   \begin{equation*}
   c_0 + c_1x + \dots + c_{n-1}x^{n-1} \in \mathbb{F}_q[x]/(x^n - 1),
   \end{equation*}
such a cyclic code can be identified with an ideal of the quotient ring $\mathbb{F}_q[x]/(x^n - 1)$. Since $\mathbb{F}_q[x]/(x^n - 1)$ is a principal ideal ring, every cyclic code $\mathcal{C}$ is generated by a unique monic divisor $g(x)$ of $x^n-1$, denoted by $\mathcal{C} = \langle g(x) \rangle$. The polynomial $g(x)$ is called the \textit{generator polynomial} of $\mathcal{C}$. 

For the case that is most interesting to us, let us assume that $n=q^m-1$. Recall that $\beta$ is a primitive element of $\mathbb{F}_{q^m}$. Thus, every root of $g(x)$ is of the form $\beta^i$. The set 
\begin{equation}\notag
 T(\mathcal{C})=  \{i: 0\leq i\leq n-1, g(\beta^i)=0\} 
\end{equation}
is called the \textit{defining set} of $\mathcal{C}$ with respect to $\beta$. 
It is well known that the dimension of $\C$ is given by
\begin{equation}\label{dim}
\dim(\C)
=
\left|\{0,1,\ldots,n-1\}\setminus T(\C)\right|,
\end{equation}
where $|\cdot|$ denotes the cardinality of a set.
Moreover, for two cyclic codes $\C$ and $\C'$, it is
straightforward to verify that
\begin{equation}\label{pre2}
\mathcal{C}\subseteq \mathcal{C}^{\prime} \quad \text{if and only if} \quad T(\mathcal{C}^{\prime}) \subseteq T(\mathcal{C}).
\end{equation}

\subsection{Cyclotomic cosets}
For two   integers $a$ and $b$ with $b\geq 1$,  we denote by  $a\bmod b$   the least nonnegative integer congruent to $a$ modulo $b$. Equivalently,
\begin{equation*}
    a\bmod b = a  - b  \left\lfloor \frac{a}{b} \right\rfloor.
\end{equation*}
For each integer $a$ with $0\leq a\leq n-1$, the \emph{$q$-cyclotomic coset} of $a$ modulo $n$ is defined by
\begin{equation*}
    C_a=\{aq^k\bmod n:k= 0,1,2,\ldots,m-1\}.
\end{equation*}
The smallest integer in $C_a$   is called the  \emph{coset leader} of $C_a$, denoted by $\cl(a)$. 
It is well known that 
   the minimal polynomial $m_a(x)$ of $\beta^a$ over $\mathbb{F}_q$ can be expressed as 
\begin{equation}\notag
    m_a(x)=\prod\limits_{i\in C_a}(x-\beta^i).
\end{equation} 
As a result, the defining set of 
$\mathcal{C}_{(q,m,\delta)}$ with respect to $\beta$ is precisely 
\begin{equation}\label{pre1}
  T(\mathcal{C}_{(q,m,\delta)}) = \bigcup_{a=1}^{\delta-1} C_a.
\end{equation}
Observe  that  for any  integer  $\delta$  satisfying $2\leq \delta\leq  n-1$, 
$
\bigcup\limits_{a=1}^{\delta-1}C_a \neq  \bigcup\limits_{a=1}^{\delta}C_a$
 if and only if 
$\delta$ is a coset leader.
  Consequently,   we have   
\begin{equation}\label{db}
d_B\left(\mathcal{C}_{(q,m,\delta)}\right)=\delta \quad \hbox{ if and only if } \quad \delta=\cl(\delta). 
\end{equation}

To facilitate our subsequent discussion, we introduce some further notation.
Let $\mathcal{D}_q=\{0,1,\ldots,q-1\}$  denote the set of all nonnegative integers less than $q$. 
 Note that 
each  integer $a$ with $0\leq a\leq n-1$ can be uniquely represented by its \textit{$q$-adic expansion} as  $a=\sum_{\ell=0}^{m-1}a_\ell q^\ell$, where  $a_\ell\in \mathcal{D}_q  
$ for $0\leq \ell\leq m-1$. We define the 
 \emph{$q$-adic vector} of $a$ by 
\begin{equation*}
    [a]_q=(a_{m-1},a_{m-2},\ldots,a_1,a_0),
\end{equation*}
and define the \textit{$q$-weight} of $a$ by 
\begin{equation}\notag
\wtq(a)=\sum_{i=0}^{m-1}a_i. 
\end{equation}

Let $
\mathcal{D}_q^m = \{(a_{m-1}, \ldots, a_0) : a_i \in \mathcal{D}_q \}
$ 
denote the set of all length-$m$ vectors over $\mathcal{D}_q$. We equip $\mathcal{D}_q^m$ with the lexicographic order. More precisely, for two distinct vectors
$
\mathbf{u}=(u_{m-1},\ldots,u_0) $ and $ \mathbf{w}=(w_{m-1},\ldots,w_0)$
in $\mathcal{D}_q^m$, let $j$ be the largest integer  such that $u_j \neq w_j$. We write $\mathbf{u}<\mathbf{w}$ if $u_j<w_j$, and $\mathbf{u}\leq\mathbf{w}$ if either $\mathbf{u}<\mathbf{w}$ or $\mathbf{u}=\mathbf{w}$.
 Let $I_{m-1}$ be the identity matrix of order $m-1$, and let 
\begin{equation*}
   Q = \begin{bmatrix}
        \mathbf{0} & 1 \\
        I_{m-1} & \mathbf{0}
    \end{bmatrix}
\end{equation*}
be an $m\times m$ matrix.
Then for any vector $(a_{m-1},a_{m-2},\ldots,a_0) \in \mathcal{D}_q^m$, we have
\begin{equation*}
    (a_{m-1},a_{m-2},\ldots,a_0)Q = (a_{m-2},\ldots,a_{0},a_{m-1}).   
\end{equation*}
The following proposition records several immediate consequences of the
definitions.
\begin{proposition}\label{properties}
Let $a$ and $b$ be two integers with $0\leq a,b\leq n-1$. Then the following
statements hold.
\begin{enumerate}
    \item  $a=b$ if and only if $[a]_q=[b]_q$, and $a<b$ if and
    only if $[a]_q<[b]_q$.
    \item 
    $
    [a q^k\bmod n]_q=[a]_qQ^k
$  for all $0\leq k\leq m-1$.
    \item  $a=\cl(a)$  if and only if
    $
    [a]_qQ^k\geq  [a]_q
   $ 
    for all $1\leq k\leq m-1$.
\end{enumerate}
\end{proposition}

\subsection{Punctured generalized Reed--Muller codes}
For a polynomial
$f\in\mathbb{F}_q[x_1,\ldots,x_m]$, let $\deg_{x_i}(f)$ denote the degree of $f$ with respect to $x_i$, and
let $\deg(f)$ denote its total degree.
The polynomial $f$ is called \emph{reduced} if
$
\deg_{x_i}(f)\leq q-1$
for all $1\leq i\leq  m$.
Let $\ell$ be an integer with $1\leq\ell<m(q-1)$. Denote by
$\mathcal{P}_q(\ell,m)$ the set of all reduced polynomials
$f\in\mathbb{F}_q[x_1,\ldots,x_m]$ satisfying $\deg(f)\leq\ell$.
The generalized Reed--Muller code of order $\ell$ is defined by
\begin{equation}\notag
\mathrm{GRM}_q(\ell,m)
=
\left\{
\bigl(f(\mathbf{x})\bigr)_{\mathbf{x}\in\mathbb{F}_q^m}:
f\in\mathcal{P}_q(\ell,m)
\right\}.
\end{equation}
The punctured generalized Reed--Muller code of order
$\ell$, denoted by $\mathrm{PGRM}_q(\ell,m)$, is obtained from
$\mathrm{GRM}_q(\ell,m)$ by deleting the coordinate indexed by
the zero vector $\mathbf{0}\in\mathbb{F}_q^m$. That is, 
\begin{equation}\notag
\mathrm{PGRM}_q(\ell,m)
=
\left\{
\bigl(f(\mathbf{x})\bigr)_{\mathbf{x}\in\mathbb{F}_q^m \setminus 
\{\mathbf{0}\}}:
f\in\mathcal{P}_q(\ell,m)
\right\}.
\end{equation}
By the standard characterization \cite{DelsarteGoethalsMacWilliams1970,DingLiXia2018},  the code  $\mathrm{PGRM}_q(\ell,m)$ can be identified with the cyclic code of length $n=q^m-1$ whose defining set with respect to $\beta$ is
\begin{equation}\label{eqdef}
T\left(\mathrm{PGRM}_q(\ell,m)\right)
=
\left\{
a:
1\leq a\leq n-1,\ 
\wtq(a)<m(q-1)-\ell
\right\}.
\end{equation}

\section{Auxiliary results}\label{sec:auxiliary}
In this section, we introduce several auxiliary lemmas required to prove our main result. The first one, due to Ax \cite{Ax1964}, relates the degree of a polynomial over a finite field to the number of its zeros. A detailed proof can also be found in \cite[Theorem 5.2]{Hou2018}. 
\begin{lemma}\textnormal{\cite{Ax1964}}\label{thmax}
Let $f\in\F_q[x_1,\ldots,x_m]$ be a polynomial with  $\deg(f)>0$, and let $N_f$ denote the total number of zeros of $f$, i.e., \begin{equation*}
N_f=\left|\{\mathbf{x}\in\F_q^m:f(\mathbf{x})=0\}\right|.
\end{equation*}
Then $q^{\lceil m/  \deg(f)\rceil-1}$ divides $N_f$. Equivalently, if $b$ is the largest integer such that $b<\frac{m}{\deg(f)}$, then $q^b$ divides $N_f$.
\end{lemma}

The following  result concerning the weights of codewords in punctured generalized Reed--Muller codes is a direct consequence of Lemma \ref{thmax}.
\begin{corollary}\label{lemma2}
Let $\ell$ be an integer satisfying $1\leq \ell<m(q-1)$. Then every codeword
$\mathbf{c}\in\PGRM_q(\ell,m)$ satisfies
\begin{equation}\label{contra}
\wt(\mathbf{c})
\equiv 0
\pmod{q^{\left\lfloor (m-1)/\ell\right\rfloor}}
\quad\text{or}\quad
\wt(\mathbf{c})
\equiv -1
\pmod{q^{\left\lfloor (m-1)/\ell\right\rfloor}}.
\end{equation}
\end{corollary}

\begin{proof}
Let $\mathbf{c}\in\PGRM_q(\ell,m)$, and choose a codeword
$\widetilde{\mathbf{c}}\in\GRM_q(\ell,m)$ whose puncturing yields
$\mathbf{c}$. Depending on whether the punctured coordinate of
$\widetilde{\mathbf{c}}$ is zero or nonzero, we have
\begin{equation*}
\wt(\widetilde{\mathbf{c}})=\wt(\mathbf{c})
\quad\text{or}\quad
\wt(\widetilde{\mathbf{c}})=\wt(\mathbf{c})+1,
\end{equation*}
respectively.
Therefore, to establish \eqref{contra}, it suffices to show that
$q^{\left\lfloor(m-1)/\ell\right\rfloor}\mid
\wt(\widetilde{\mathbf{c}})$.

By definition, there exists a polynomial
$f\in\F_q[x_1,\ldots,x_m]$ with $\deg(f)\leq\ell$ such that
\begin{equation*}
\widetilde{\mathbf{c}}
=
\bigl(f(\mathbf{x})\bigr)_{\mathbf{x}\in\mathbb{F}_q^m}.
\end{equation*}
It follows that
\begin{equation}\label{wt}
\wt(\widetilde{\mathbf{c}})=q^m-N_f,
\end{equation}
where
$N_f=|\{\mathbf{x}\in\F_q^m:f(\mathbf{x})=0\}|$.

If $\deg(f)\leq 0$, then
$\wt(\widetilde{\mathbf{c}})=q^m$ or
$\wt(\widetilde{\mathbf{c}})=0$. In either case,  
$q^{\left\lfloor(m-1)/\ell\right\rfloor}\mid
\wt(\widetilde{\mathbf{c}})$.

Now suppose that $1\leq\deg(f)\leq\ell$. By Lemma~\ref{thmax}, we obtain 
\begin{equation}\label{lema2e1}
q^{\lceil m/\deg(f)\rceil-1}\mid N_f.
\end{equation}
Since $\deg(f)\leq\ell$, we have
\begin{equation}\label{lema2e2}
\left\lceil\frac{m}{\deg(f)}\right\rceil-1
\geq
\left\lceil\frac{m}{\ell}\right\rceil-1
=
\left\lfloor\frac{m-1}{\ell}\right\rfloor.
\end{equation}
It follows from \eqref{wt}, \eqref{lema2e1}, and \eqref{lema2e2} that
$q^{\left\lfloor(m-1)/\ell\right\rfloor}\mid
\wt(\widetilde{\mathbf{c}})$. Thus, this divisibility holds in all cases,
and hence \eqref{contra} follows.
\end{proof}

We shall also use the following result of Kasami and Lin \cite{KasamiLin1972} to determine
the exact minimum distances of the binary BCH codes considered in this paper.

\begin{lemma}\textnormal{\cite{KasamiLin1972}}\label{kl}
Let $m$, $i$, and $j$  be integers with 
$
1\leq i\leq m-j-2$ and 
$0\leq j\leq m-2i.
$ Let $\delta=2^{m-1-j}-2^{m-1-j-i}-1.
$ 
Then
$\mathcal{C}_{(2,m,\delta)}$ satisfies
\begin{equation*}
d\bigl(\mathcal{C}_{(2,m,\delta)}\bigr)=\delta.
\end{equation*}
\end{lemma}

\section{Proof of the main result}\label{mainresult}

With the  preparations provided in the previous sections, we are now ready to prove Theorem \ref{mainth}.

\begin{proof}[Proof of Theorem \ref{mainth}]
The proof is divided into four steps: we first show that $d_B(\mathcal{C}_{(q,m,\delta)})=\delta$, i.e., $\delta$ is a coset leader. Next, we show that $\mathcal{C}_{(q,m,\delta)} \subseteq
\mathrm{PGRM}_q(3,m)$. Then by combining $d(\mathcal{C}_{(q,m,\delta)}) \ge \delta$ and Corollary \ref{lemma2}, we prove \eqref{1:gapb} of Theorem \ref{mainth}. Now Finally we deal with the case that $q=2$. 

{\bf Step 1.} We first prove that $d_B(\mathcal{C}_{(q,m,\delta)})=\delta$, that is, $\delta$ is a coset leader.  
By Proposition~\ref{properties}, it suffices to prove that
\begin{equation}\label{suf}
[\delta]_qQ^k>[\delta]_q
\quad\text{for all }1\leq k\leq m-1.
\end{equation}

A direct calculation
gives the $q$-adic vector of $\delta$ as 
\begin{equation}\label{deltaq}
    [\delta]_q = \left( q-2, (q-1)^{(u-1)}, q-2, (q-1)^{(b)}, q-2, (q-1)^{(s)} \right),
\end{equation}
where $b = m - u - s - 2$ and $x^{(r)}$ denotes a block consisting of $r$ consecutive copies of
$x$.
The assumptions $u = \lfloor m/4 \rfloor$ and $m \ge 13$ imply that
\begin{equation*}
    m - 1 \le 4u + 2 < 6u.
\end{equation*}
It follows that
\begin{equation*}
    t = \left\lfloor \frac{m-1}{3} \right\rfloor \le 2u - 1.
\end{equation*}
Since $s < t$, we obtain
\begin{equation}\label{hao}
    s \le t - 1 \le 2u - 2.
\end{equation}
Consequently,
\begin{equation}\label{bgequ}
\begin{split}
    b &= m - u - s - 2 \\
    &\ge m - 3u\\ 
    &\ge u.
    \end{split}
\end{equation} 
Therefore, the first $u+1$ digits of $[\delta]_qQ^u$ and $[\delta]_q$
are
\begin{equation*}
    \left(q-2,(q-1)^{(u)}\right)\quad \hbox{and}\quad 
\left(q-2,(q-1)^{(u-1)},q-2\right), 
\end{equation*}
respectively. It follows that $[\delta]_q Q^{u} > [\delta]_q$. 
Similarly, since $s \ge u$, the first $u+1$ digits of
$[\delta]_q Q^{u+b+1}$ are also  $\left( q-2, (q-1)^{(u)} \right)$. Hence,
$
    [\delta]_q Q^{u+b+1} > [\delta]_q.
$

Moreover, for each integer  $k$ satisfying $1\leq k\leq u-1$,
$u+1\leq k\leq u+b$, or $u+b+2\leq k\leq m-1$, 
the first digit of $[\delta]_q Q^k$ is $q-1$, whereas the first digit of
$[\delta]_q$ is $q-2$.
This implies that
\begin{equation*}
    [\delta]_q Q^k > [\delta]_q
\end{equation*}
for all such $k$. Together with the cases $k=u$ and $k=u+b+1$ considered above, this
proves \eqref{suf} and hence establishes
$d_B(\mathcal{C}_{(q,m,\delta)})=\delta$.

{\bf Step 2.} Now we prove that
\begin{equation}\label{contain}
\mathcal{C}_{(q,m,\delta)}
\subseteq
\mathrm{PGRM}_q(3,m).
\end{equation}
By \eqref{eqdef},  the defining set of $\mathrm{PGRM}_q(3,m)$  is 
\begin{equation}\notag
T(\mathrm{PGRM}_q(3,m)) = \{ a \in \{ 1, \dots, n-1\} : \mathrm{wt}_q(a) < m(q-1)-3 \}.
\end{equation}
Let $a \in T(\mathrm{PGRM}_q(3,m))$ and let $[a]_q = (a_{m-1}, a_{m-2}, \dots, a_0)$ be its $q$-adic vector. Then we have
\begin{equation}\label{inequ}
    \sum_{i=0}^{m-1} (q - 1 - a_i) = m(q-1) - \mathrm{wt}_q(a) > 3.
\end{equation}
We now establish
\begin{equation*}
   [\cl(a)]_q < [\delta]_q.
\end{equation*}
by considering the following two cases.

First, suppose that $a_i \le q - 3$ for some integer $i$ with $0\leq i \leq  m-1$. Note that $a_i$ is precisely the first digit of $[a]_q Q^{m-1-i}$, while the first digit of $[\delta]_q$ is $q-2$. Thus, we have
\begin{equation*}
    [\cl(a)]_q \le [a]_q Q^{m-1-i} < [\delta]_q.
\end{equation*}

Now suppose that $a_i \ge q - 2$ for all integers $i$ with $0\leq i\leq m-1 $. Let
\begin{equation*}
k = |\{i : 0\leq i\leq m-1, a_i = q-2\}|
\end{equation*}
denote the number of digits of $[a]_q$  equal to $q-2$. It follows from \eqref{inequ} that $k \ge 4$.

If there exist two cyclically consecutive entries equal to $q-2$ in the $q$-adic vector $[a]_q$ of $a$, then there exists some integer $j$ with $0\leq j \leq   m-1$ such that the first two digits of $[a]_q Q^j$ are both equal to $q-2$. Observe that the first two digits of $[\delta]_q$ are $q-2$ and $q-1$, respectively.
It follows  that
\begin{equation*}
    [\mathrm{cl}(a)]_q \le [a]_q Q^j < [\delta]_q.
\end{equation*}

Otherwise, after replacing $[a]_q$ by $[a]_qQ^j$ for a suitable integer
$j$, if necessary, we may assume that the $q$-adic vector of $a$ is of the form
\begin{equation*}
    [a]_q = \left( q-2, (q-1)^{(g_1)}, q-2, (q-1)^{(g_2)}, \dots, q-2, (q-1)^{(g_k)} \right),
\end{equation*}
where $g_i \ge 1$ for all $1 \le i \le k$, and $g_1 = \min \{g_i : 1 \le i \le k\}$. It is clear that
\begin{equation*}
    \sum_{i=1}^k g_i = m - k.
\end{equation*}
Therefore, we have
\begin{equation*}
    g_1  \le \left\lfloor \frac{m-k}{k} \right\rfloor \le \left\lfloor \frac{m-4}{4} \right\rfloor = u - 1.
\end{equation*}

If $g_1<u-1$, then the first $g_1+2$ digits of $[a]_q$ and 
$[\delta]_q$ are 
\begin{equation*}
    \left(q-2,(q-1)^{(g_1)},q-2\right)\quad \hbox{and}\quad 
\left(q-2,(q-1)^{(g_1+1)}\right),
\end{equation*}
respectively. It follows that
\begin{equation} \label{hhh}  
[\mathrm{cl}(a)]_q \le [a]_q < [\delta]_q.
\end{equation}

It remains to consider the case $g_1=u-1$. Since $g_1 = \min \{g_i : 1 \le i \le k\}$, we have $g_i \ge u - 1$ for all $1 \le i \le k$. Therefore,
\begin{align*}
    g_2 &= m - k - \sum_{\substack{1\leq i\leq k\\ i\neq2}}g_i\\
       & \le m - k - (k-1)(u-1) \\
       &\le m - 1 - 3u,
\end{align*}
where the last inequality follows from $k \ge 4$. On the other hand, it follows from \eqref{hao} that
\begin{equation}\label{inequality}
    b - (m - 1 - 3u) = 2u - 1 - s \ge 1.
\end{equation}
Combining the above inequalities, we obtain $g_2 < b$. Then, by comparing the first $g_1 + g_2 + 3$ digits of $[a]_q$ and $[\delta]_q$, we obtain \eqref{hhh} again.

We have thus shown that $[\mathrm{cl}(a)]_q < [\delta]_q$ in all cases. Recalling \eqref{pre1} and $\mathrm{cl}(\delta) = \delta$, it follows that $a \in T(\mathcal{C}_{(q,m,\delta)})$. Therefore,
\begin{equation}\notag
T\bigl(\mathrm{PGRM}_q(3,m)\bigr)
\subseteq
T\bigl(\mathcal{C}_{(q,m,\delta)}\bigr).
\end{equation}
With \eqref{pre2}, this implies that \eqref{contain} holds.

{\bf Step 3.} We now proceed to prove   
$d(\mathcal{C}_{(q,m,\delta)}) \geq \delta+q^{s}$.
Suppose that $w = \mathrm{wt}(\mathbf{c})$ for some nonzero codeword $\mathbf{c} \in \mathcal{C}_{(q,m,\delta)}$. Then we first have $w \ge \delta$.
By \eqref{contain}, $\mathbf c$ is also a 
codeword of $\PGRM_q(3,m)$. Therefore,  Corollary~\ref{lemma2}
gives 
\begin{equation}\notag
    w \bmod q^t = 0 \quad \text{or} \quad w \bmod q^t = q^t - 1,
\end{equation}
where $t=\lfloor(m-1)/3\rfloor$.
 If $w \le \delta + q^s - 1$, then since $w\geq \delta$,  $t \le m - 1 - u$ and $s < t \le m - 1$, we have
\begin{equation*}
    q^t - q^s - 1 \le w \bmod q^t \le q^t - 2.
\end{equation*}
This contradicts \eqref{contra}. Therefore, we must have $w \ge \delta + q^s$, which implies that $d(\mathcal{C}_{(q,m,\delta)}) \geq \delta+q^{s}$. This proves \eqref{1:gapb} of Theorem \ref{mainth}. 

{\bf Step 4.} Finally, we prove that the lower bound on the minimum distance is attained
when $q=2$, i.e.,
\begin{equation}\label{maineq3}
d(\mathcal{C}_{(2,m,\delta)}) = \delta+{2}^s,
\end{equation}
where $\delta=2^{m-1}-2^{m-1-u}-2^s-1$. 

Since $m\geq 13$ and $u=\left\lfloor m/4\right\rfloor$, we have
$2\leq u\leq m-2$ and $m-2u\geq0$. Notice that $\delta+2^s=2^{m-1}-2^{m-1-u}-1$. 
By applying Lemma \ref{kl} with $(i,j)=(u,0)$, we obtain 
\begin{equation*}
d(\mathcal{C}_{(2,m,\delta+2^s)}) = \delta+2^s.
\end{equation*}
By \eqref{pre1},
\begin{equation*} T(\mathcal{C}_{(2,m,\delta)})  \subseteq T(\mathcal{C}_{(2,m,\delta+2^s)}), 
\end{equation*}
and hence  by \eqref{pre2}, 
\begin{equation*}  \mathcal{C}_{(2,m,\delta+2^s)} \subseteq \mathcal{C}_{(2,m,\delta)}.
\end{equation*}
It follows that 
\begin{equation}\notag
d(\mathcal{C}_{(2,m,\delta)})
\leq d(\mathcal{C}_{(2,m,\delta+2^s)})
= \delta+2^s. 
\end{equation}
Combining this with the lower bound already established yields \eqref{maineq3}.
\end{proof}

\section{Conclusion and Remarks}\label{conclusion}
In this paper, we have constructed an infinite family of primitive
narrow-sense BCH codes for which the gap between the minimum distance
and the Bose distance is unbounded. This family disproves Charpin's
conjecture and shows that the proximity between $d$ and $d_B$ observed
in previously known examples does not reflect a general structural
property of BCH codes. 

For the subfamily corresponding to $s=t-1$ in
Theorem~\ref{mainth}, we have
\begin{equation*}
d-d_B\geq q^{\lfloor(m-1)/3\rfloor-1}.
\end{equation*}
For each fixed prime power $q$, this lower bound is
$\Theta(n^{1/3})$, where $n=q^m-1$ is the code length. It is natural to ask whether there
exist infinite families for which the gap grows faster and, more
generally, to determine the largest possible value of $d-d_B$ among
primitive narrow-sense BCH codes of a given length.

\section*{Acknowledgment}
The first author gratefully acknowledges financial support from Professor Cunsheng Ding. The authors acknowledge the use of  ChatGPT for exploratory discussions during the early phase of the research.  All mathematical results were developed, proved, and independently verified by the authors, who take full responsibility for the content and integrity of this work.

\section{Appendix} \label{append}

In this Appendix, we prove that the dimension of
$\mathcal{C}_{(q,m,\delta)}$ is given by \eqref{dimeq}. 

\begin{proof} Since $\cl(\delta)=\delta$,
\eqref{pre1} gives 
\begin{equation}\label{sets}
\{0,1,\ldots,n-1\}\setminus T(\mathcal{C}_{(q,m,\delta)})
=
\{0\}\cup
\{a:1\leq a\leq n-1,\ \cl(a)\geq\delta\}.
\end{equation}
Therefore, by \eqref{dim}, it suffices to count the integers in the set
\begin{equation*}
\mathcal{A}
=
\{a:1\leq a\leq n-1,\ \cl(a)\geq\delta\}.
\end{equation*}

Let $a\in\mathcal{A}$. Recalling the $q$-adic vector of $\delta$ given in
\eqref{deltaq}, we see that every digit of $[a]_q$ belongs to
$\{q-2,q-1\}$, and $[a]_q$ cannot contain two cyclically consecutive
digits equal to $q-2$.  Indeed, if either condition failed, there would exist an
integer $i$ with $1\leq i\leq m$ such that
$[\cl(a)]_q\leq [a]_qQ^i<[\delta]_q$, contradicting $\cl(a)\geq \delta$.   

We first claim that $[a]_q$ contains at most three digits equal to $q-2$.
Suppose, to the contrary, that it contains $v\geq4$ such digits. Then
\begin{equation*}
[\cl(a)]_q
=
\bigl(
q-2,(q-1)^{(h_1)},q-2,(q-1)^{(h_2)},
\ldots,q-2,(q-1)^{(h_v)}
\bigr),
\end{equation*}
where $h_1=\min\{h_i:1\leq i\leq v\}$,
$\sum_{i=1}^{v}h_i=m-v$, and $h_i\geq1$ for every integer  $i$ with $1\leq i\leq v$. Therefore,
\begin{equation*}
h_1
\leq
\left\lfloor\frac{m-v}{v}\right\rfloor
\leq
\left\lfloor\frac{m-4}{4}\right\rfloor
=u-1.
\end{equation*}
The condition  $\cl(a)\geq\delta$ implies $h_1\geq u-1$. Consequently, 
$h_1=u-1$ and  $h_i\geq u-1$ for every integer $i$ with  $1\leq i\leq v$. 
It follows that
\begin{equation*}
\begin{aligned}
h_2
&=m-v-\sum_{\substack{1\leq i\leq v\\i\neq2}}h_i\\
&\leq m-v-(v-1)(u-1)\\
&\leq m-3u-1\\
&<b,
\end{aligned}
\end{equation*}
where the last two inequalities follow from $v\geq4$ and
\eqref{inequality}, respectively.  Together with $h_1=u-1$, this implies that  $[\cl(a)]_q<[\delta]_q$, contradicting
$\cl(a)\geq\delta$. Thus, the claim follows.

Since $a\leq n-1=q^m-2$, the vector $[a]_q$ cannot consist entirely of
digits equal to $q-1$. It therefore contains exactly one, two, or three
digits equal to $q-2$. For $j\in\{1,2,3\}$, let $N_j$ denote the number
of integers $a\in\mathcal{A}$ such that  $[a]_q$ contains exactly $j$
digits equal to $q-2$. Then we distinguish the following cases.

\emph{Case 1:} The vector $[a]_q$ has exactly one digit equal to $q-2$.
 There are $m$ possible positions for this digit, and each resulting vector satisfies
$[\cl(a)]_q>[\delta]_q$ by \eqref{deltaq}. Hence, we have
\begin{equation*}
    N_1=m.
\end{equation*}

\emph{Case 2:} The vector $[a]_q$ contains exactly two digits equal to
$q-2$. By \eqref{deltaq}, the condition $\cl(a)\geq\delta$ is
equivalent to requiring that, when the entries of $[a]_q$ are read
cyclically, at least $u-1$ digits lie between the two digits equal to
$q-2$ in each direction.

There are $m$ choices for the first digit equal to $q-2$. Once it is
chosen, the other digit equal to $q-2$ cannot lie among the $u-1$
entries immediately preceding it or the $u-1$ entries immediately
following it. Hence, there are $m-1-2(u-1)=m-2u+1$ choices for the
other digit. 
Since either digit equal to $q-2$ may have been chosen first, the
number of integers in this case is
\begin{equation*}
N_2=\frac{m(m-2u+1)}{2}.
\end{equation*}
 
\emph{Case 3:}  The vector $[a]_q$ contains exactly three digits equal to
$q-2$. Fix one of these digits and read the entries cyclically from it.
Let  $h_1,h_2,h_3$
denote the numbers of digits between successive digits equal to
$q-2$.   For a fixed first digit, the choices of the other two digits equal to
$q-2$ are in one-to-one correspondence with the ordered triples $(h_1,h_2,h_3)$  satisfying $h_1+h_2+h_3=m-3$.  

By \eqref{deltaq},
 the condition $\cl(a)\geq\delta$ is equivalent to
the following two conditions:
\begin{enumerate}
\item[(i)]  $h_i\geq u-1$ for all $1\leq i\leq3$.
\item[(ii)] If $h_1=u-1$, then $h_2\geq b$; if $h_2=u-1$, then
$h_3\geq b$; and if $h_3=u-1$, then $h_1\geq b$.
\end{enumerate}

We first consider the triples satisfying $h_i\geq u$ for all
$1\leq i\leq3$. The substitution $x_i=h_i-u$ gives a bijection between
these triples and the nonnegative integer solutions of
$x_1+x_2+x_3=m-3-3u$. Hence,  the number of
such triples is
$
\binom{m-3u-1}{2}.
$

It remains to consider the triples for which some $h_i$ equals $u-1$.
Condition (ii), together with $b\geq u$ established in  \eqref{bgequ}, implies that exactly one
of the $h_i$ equals $u-1$. There are three choices for this index.
Suppose, for example, that $h_1=u-1$. Then $h_2\geq b$ and  $h_3\geq u$. Since $h_1+h_2+h_3=m-3$, it follows that  
\begin{equation*}
    h_3=m-3-(u-1)-h_2,
\end{equation*}
and 
\begin{equation*}
b\leq h_2\leq m-3-(u-1)-u.
\end{equation*}
Therefore, the number of choices for
$(h_2,h_3)$ is
\begin{equation*}
m-2u-2-b+1=s-u+1.
\end{equation*}
Thus, this case gives $3(s-u+1)$ triples.

Combining these two counts, for each fixed first digit equal to $q-2$,
the number of choices for the other two digits is
$\binom{m-3u-1}{2}+3(s-u+1)$. There are $m$ choices for the fixed
digit, and any of the three digits equal to $q-2$ may serve as this
digit. Hence, the number of integers in this case is
\begin{equation*}
N_3
=
\frac{m}{3}
\left(
\binom{m-3u-1}{2}+3(s-u+1)
\right).
\end{equation*}

Combining the three cases, we obtain
\begin{equation*}
\begin{aligned}
|\mathcal{A}|
&=N_1+N_2+N_3\\
&=
m+\frac{m(m-2u+1)}{2}
+\frac{m}{3}
\left(
\binom{m-3u-1}{2}+3(s-u+1)
\right).
\end{aligned}
\end{equation*}
Together with \eqref{dim} and \eqref{sets}, this yields \eqref{dimeq}. This completes the proof of \eqref{dimeq}.
\end{proof}

\end{document}